\documentclass{article}
\usepackage{graphicx} 
\usepackage{bm}
\usepackage{amsfonts}
\usepackage{amsmath}
\usepackage{amssymb}

\usepackage{xcolor}

\usepackage[
backend=biber,
style=numeric,
sorting=ynt    
]{biblatex}
\usepackage{amsthm}
\newtheorem{theorem}{Theorem}
\newtheorem{corollary}{Corollary}
\newtheorem{proposition}{Proposition}
\newtheorem{lemma}{Lemma}
\theoremstyle{definition}
\newtheorem{definition}{Definition}
\theoremstyle{remark}

\newcommand{\uu}{{\bm u}}
\newcommand{\vv}{{\bm v}}
\newcommand{\x}{{\bm x}}
\newcommand{\X}{{\bm X}}
\newcommand{\y}{{\bm y}}
\newcommand{\Y}{{\bm Y}}
\newcommand{\z}{{\bm z}}

\title{Algorithms for Approximating Conditionally Optimal Bounds}
\author{George Bissias}
\date{March 2025}

\begin{document}

\maketitle

\section{Preliminaries}

In this document we consider lower bounds on discrete distributions $\mathcal{F}^m$ over support set 
\begin{equation}
    S^m = \left\{ S_{\min} + i\frac{S_{\max}-S_{\min}}{m-1} ~\bigg|~ i \in \mathbb{N}, 0 \leq i \leq m-1 \right \}.     
\end{equation}
Samples of size $n$ are drawn from sample space $\Omega^m = S^m \times \ldots \times S^m$. We further assume that samples are unique when their components are arranged in increasing order. When a statement is true for any $m$, we often drop the superscript from the associated quantities.

Our results concern bounds \emph{consistent} with \emph{preorders} on the samples in $\Omega$. A preorder can be characterized by a set of equivalence classes, combined with a partial order over those classes. In the case where every equivalence class contains a single sample, and the classes are totally ordered, the preorder is also a total order. Learned-Miller~\cite{learnedmiller2025admissibilityboundsmeandiscrete} formalized the notion of consistency for total orders and we extend that theory to preorders. 

\begin{definition}
\label{def:preorder}
    For any given preorder $R$ and samples $\x, \y \in \Omega$, we write $\x \lesssim_R \y$ whenever $\y$ is greater than or equal to $\x$ in $R$. A preorder is characterized by the following properties
    \begin{itemize}
        \item $\x \lesssim_R \x$.
        \item $\x \lesssim_R \y$ and $\y \lesssim_R \bm z$ implies that $\x \lesssim_R \bm z$.
    \end{itemize}
    When $\x \lesssim_R \y$ and $\y \lesssim_R \x$ we write $\x \sim_R \y$. And by $\x <_R \y$ we indicate that $\x \lesssim_R \y$ and $\x \not \sim_R \y$. 
\end{definition}

\begin{definition}
\label{def:partial_order}
    \emph{Partial order} $R$ is a preorder where $\forall \x, \y \in \Omega, \x \sim_R \y$ implies that $\x = \y$. When $\x$ is no greater than $\y$ in partial order $R$ we write $\x \leq_R \y$. We further user $\x <_R \y$ to indicate that $\x \leq_R \y$ and $\x \neq \y$.
\end{definition}

\begin{definition}
\label{def:total_order}
    \emph{Total order} $T$ is a partial order that is \emph{strongly connected}, i.e., $\forall \x, \y \in \Omega, \x \leq_T \y$ or $\y \leq_T \x$.
\end{definition}

\begin{definition}
\label{def:total_preorder}
    \emph{Total preorder} $R$ is a preorder that is strongly connected, i.e. where $\forall \x, \y \in \Omega, \x \lesssim_R \y$ or $\y \lesssim_R \x$. A total preorder is a total order over equivalence classes.
\end{definition}

\begin{proposition}
\label{prop:less_implies_not_greater_equal}
    If $R$ is a preorder, then for any $\x, \y \in \Omega$. We have
    \begin{equation}
        \x <_R \y \Rightarrow \neg (\y \lesssim_R \x).
    \end{equation}
    If additionally $R$ is a total preorder, then we also have
    \begin{equation}
        \neg (\y \lesssim_R \x) \Rightarrow \x <_R \y.
    \end{equation}
\end{proposition}

\begin{proof}
To prove the forward direction, suppose for the sake of contradiction that $\x <_R \y$ and $\y \lesssim_R \x$. Since $\x <_R \y$ we know that $\x \lesssim_R \y$ and $\x \not \sim_R \y$. But since $\x \lesssim_R \y$ and $\y \lesssim_R \x$ we have by Definition~\ref{def:preorder} that $\x \sim_R \y$, which contradicts the earlier observation that $\x \not \sim_R \y$.

To prove the reverse direction, suppose that $R$ is a total preorder and $\neg (\y \lesssim_R \x)$. Since $R$ is a total preorder, we know by Definition~\ref{def:total_preorder} that either $\x \lesssim_R \y$ or $\y \lesssim_R \x$. Therefore, it must be the case that $\x \lesssim_R \y$. Now Definition~\ref{def:preorder} requires that $\x \lesssim_R \y$ \emph{and} $\y \lesssim_T \x$ in order for $\y \sim_R \x$. Therefore $\x \not \sim_R \y$, since $\neg (\y \lesssim_R \x)$. Thus, $\x \lesssim_R \y$ and $\x \not \sim_R \y$ so that $\x <_R \y$ by Definition~\ref{def:preorder}.
\end{proof}

\begin{definition}
\label{def:pre_upper}
    The \emph{upper set} associated with sample $\bm x \in \Omega$ and total preorder $R$ is given by 
    \begin{equation}
        \Omega(\bm x, R) = \{\y \in \Omega: \bm x \lesssim_R \bm y\}.
    \end{equation}
\end{definition}

Note that, according to these definitions, $\forall \bm x, \bm y \in \Omega$, total preorder $R$, and total order $T$, $\bm x \lesssim_R \bm y \Leftrightarrow \bm y \in \Omega(\bm x, R)$ and $\x \leq_T \bm y \Leftrightarrow \bm y \in \Omega(\bm x, T)$.

\begin{definition}
\label{def:sample_ineq}
    For $\bm x, \bm y \in \Omega$ we write $\bm x \leq \bm y$ whenever $\forall i \in [n], x_{(i)} \leq y_{(i)}$. We write $\bm x < \bm y$ whenever $\bm x \neq \bm y \wedge \bm x \leq \bm y$.
\end{definition}

Definition~\ref{def:sample_ineq} implies that $\neg (\bm x \leq \bm y) \Leftrightarrow \exists i \in [n]: x_{(i)} > y_{(i)}$, and $\neg (\bm x < \bm y) \Leftrightarrow (\bm x = \bm y) \vee \neg (\bm x \leq \bm y)$. Furthermore, 
\begin{equation}
    \bm x < \bm y \Rightarrow \exists i \in [n]: x_{(i)} < y_{(i)} \Rightarrow \neg(\bm y \leq \bm x).
\end{equation}

\begin{definition}
    By $\bm S_i^m$ we denote the \emph{homogeneous} sample of length $n$:\\ $(S_i^m, \ldots, S_i^m) \in \Omega^m$ for $S_i^m \in S^m$. We will frequently drop the superscript $m$ from $\bm S^m_i$ when it appears in a statement that is true for any $m$.
\end{definition}

\subsection{Common orders}

\begin{definition}
\label{def:lexi_low}
    The \emph{low} lexicographic order, $T_\ell$, is the total order such that $\forall \bm x, \bm y \in \Omega, \bm x \leq_{T_\ell} \bm y \Leftrightarrow L_1(\bm x, \bm y) = 1$ where $\forall i \in [n]$,
    \begin{equation}
    \label{eq:lexi_low}
        L_i(\bm x, \bm y) = (x_{(i)} < y_{(i)}) \vee (x_{(i)} = y_{(i)}) \wedge L_{i+1}(\bm x, \bm y),
    \end{equation}
    and $L_{n+1}(\bm x, \bm y) = 1$.
\end{definition}

\begin{definition}
    The \emph{high} lexicographic order, $T_h$, is the total order such that $\forall \bm x, \bm y \in \Omega, \bm x \leq_{T_h} \bm y \Leftrightarrow H_n(\bm x, \bm y) = 1$ where $\forall i \in [n]$,
    \begin{equation}
        H_i(\bm x, \bm y) = (x_{(i)} < y_{(i)}) \vee (x_{(i)} = y_{(i)}) \wedge H_{i-1}(\bm x, \bm y),
    \end{equation}
    and $H_0(\bm x, \bm y) = 1$.
\end{definition}

\begin{definition}
\label{def:quantile_order}
    For fixed $i \in \{1, \ldots, n\}$ and any two samples $\x, \y \in \Omega$, the $i$th \emph{quantile preorder} $R_i$ is the total preorder defined by
    \begin{equation}
        \x \lesssim_{R_i} \y \Leftrightarrow x_{(i)} \leq y_{(i)}.
    \end{equation}
\end{definition}
Notice that when $n$ is odd and $i = \lceil n/2 \rceil$, the $i$th quantile is equal to the sample median.

\begin{definition}
\label{def:monotone_prop}
    Preorder $R$ is \emph{weakly monotone} if $\bm x \leq \bm y \Rightarrow \bm x \lesssim_R \bm y$, and it is \emph{strongly monotone} if $\bm x < \bm y \Rightarrow \bm x <_R \bm y$.
\end{definition}

\begin{proposition}
    Every strongly monotone order is also weakly monotone.
\end{proposition}

\begin{proof}
    Suppose that $R$ is a strongly monotone order and take any $\x, \y \in \Omega$ such that $\x \leq \y$. In the case that $\x < \y$ we know that $\x <_R \y$ as well by the strong monotonicity property, which implies that $\x \lesssim_R \y$. Now suppose that $\x = \y$. Then by the reflexive property of preorders, we have that $\x \sim_R \y$, so that $\x \lesssim_R \y$. In either case, $\x \leq \y \Rightarrow \x \lesssim_R \y$.
\end{proof}

\begin{proposition}
    A weakly monotone partial order is also strongly monotone.
\end{proposition}

\begin{proof}
    Let $R$ be an arbitrary weakly monotone partial order. By Definitions~\ref{def:partial_order} and~\ref{def:sample_ineq}  we have that 
    \begin{equation}
        \bm x < \bm y \Rightarrow (\x \leq \bm y) \wedge (\x \neq \y) \Rightarrow (\x \leq_R \bm y) \wedge (\x \neq \y) \Rightarrow (\x <_R \bm y).
    \end{equation}
\end{proof}

\begin{proposition}
    A weakly monotone preorder is strongly monotone iff $\forall \x, \y \in \Omega, \x < \y \Rightarrow \x \not \sim_R \y$.
\end{proposition}

\begin{proof}
    Let $R$ be an arbitrary weakly monotone preorder. It's clear that the conclusion is necessary for strong monotonicity. For suppose that it is false, i.e. there exists $\x, \y \in \Omega$ where $\x < \y$ and $\x \sim_R \y$, then $R$ cannot be strongly monotone since we could not have $\x <_R \y$. 
    
    To show that the conclusion is sufficient for strong monotonicity notice that, since $R$ is weakly monotone by assumption,  we have that 
    \begin{equation}
        \bm x < \bm y \Rightarrow \x \leq \bm y \Rightarrow \x \lesssim_R \bm y.
    \end{equation}
    And also by assumption, $\x \not \sim_R \y$, which implies that $\x <_R \y$. Thus, $\x < \y \Rightarrow \x <_R \y$.
\end{proof}

All monotone orders discussed in this document are strongly monotone. For brevity, we refer to such orders as simply \emph{monotone} going forward. 

\section{Topology}

Our results rely on the topological characteristics of the space of probability distributions supported on $S$. Because we deal exclusively with finite support, we may characterize the set of all distributions having support $S$ by the simplex of vectors $\Delta$, which are those ordered vectors in $[0,1]^m$ having unit L1 norm; for $\uu \in \Delta$, $u_i$ dictates the probability mass assigned to $S_i$. The \emph{natural} topology on this space is the one induced by the Euclidean distance metric: $\forall \uu, \vv \in \Delta, d(\uu, \vv) = \|\uu - \vv\|_2$.

\begin{definition}
\label{def:sets_of_dists}
    For each $\alpha \in [0, 1)$ and upper set $\Omega' \subseteq \Omega$, the \emph{interior likely set} is given by
    \begin{equation}
        \mathcal{G}(\Omega', \alpha) = \{F \in \mathcal{F}: P_F[\Omega'] > \alpha\},
    \end{equation}
    and the \emph{likely set} is given by
    \begin{equation}
        \mathcal{F}(\Omega', \alpha) = \texttt{cl}(\mathcal{G}(\Omega', \alpha)),
    \end{equation}
    where \texttt{cl} denotes the closure of a set.
\end{definition}

Let $\bm s$ denote the ordered vector of support points $S$ and note that for each $\uu \in \Delta$, corresponding to distribution $H_{\uu}$, $E[H_{\uu}] = \uu \cdot \bm s$. Thus, for $\uu, \vv \in \Delta$, corresponding to distributions $H_{\uu}$ and $H_{\vv}$ and such that $d(\uu, \vv) < \delta$, we have that 
    \begin{equation}
    \label{eq:mean_bound}
        \begin{array}{rcl}
        |E[H_{\uu}] - E[H_{\vv}]| &\leq& S_{\max} \left|\sum_{i=1}^m u_i - v_i\right| \\
        &\leq& S_{\max} \| \uu - \vv \|_1 \\
        &\leq& S_{\max} \sqrt{m} \| \uu - \vv\|_2 \\
        &=& d(\uu, \vv) \sqrt{m} S_{\max} \\
        &<& \delta \sqrt{m} S_{\max} \\
        \end{array}.
    \end{equation}

\begin{lemma}
\label{lem:mean_eps_close}
    Fix $\alpha \in [0, 1)$, $\Omega' \subseteq \Omega$, and $\delta > 0$. For each $F \in \mathcal{F}(\Omega', \alpha)$ there exists a $G \in \mathcal{G}(\Omega', \alpha)$ such that $|E[F] - E[G]| < \epsilon$.  
\end{lemma}
\begin{proof}
    Take any $F \in \mathcal{F}(\Omega', \alpha)$. By the definition of closure  we know that $\forall \delta > 0, \exists G \in \mathcal{G}(\Omega', \alpha): d(G, F) < \delta$. Choose $\delta = \epsilon / (\sqrt{m} S_{\max})$. According to Inequality~\ref{eq:mean_bound}, we have $|E[F] - E[G]| < \epsilon$.
\end{proof}

\begin{lemma}
\label{lem:closure_mean}
    Take any $\alpha \in [0, 1)$, $\mu \in [S_{\min}, S_{\max}]$, and $\Omega' \subseteq \Omega$. Then there exists $F \in \mathcal{F}(\Omega', \alpha)$ such that $E[F] = \mu$ if and only if there exists a sequence of distributions in $\mathcal{G}(\Omega', \alpha)$ that converge in mean to $\mu$.  
\end{lemma}
\begin{proof}
    Suppose first that there exists $F \in \mathcal{F}(\Omega', \alpha)$ such that $E[F] = \mu$ and choose any $\epsilon > 0$. 
    By the definition of closure  we know that $\forall i > 0, \exists G_i \in \mathcal{G}(\Omega', \alpha): d(G_i, F) < 1/i$. 
    According to Inequality~\ref{eq:mean_bound}, when $i > \sqrt{m} S_{\max} / \epsilon$, $|E[F] - E[G_i]| < \epsilon$. Thus, allowing $\epsilon$ to approach 0, we can see that the sequence of distributions $\{G_i\}_{i>0}$ converges in mean to $F$.

    Next, suppose that there exists a sequence of distributions $\{G_i\}_{i > 0}, G_i \in \mathcal{G}(\Omega', \alpha)$ converging in mean to $\mu$. It suffices for the limit of this sequence to be in $\mathcal{F}(\Omega', \alpha)$, which is true because $\lim_{i \rightarrow \infty} \{G_i\}_{i > 0}$ is a limit point of $\mathcal{G}(\Omega', \alpha)$, and the closure of any set contains all of its limit points.
\end{proof}

\section{Conditional-Optimality}

\begin{definition}
\label{def:agreement}
    Preorder order $D$ \emph{agrees} with a preorder $R$ if $\forall \x, \y \in \Omega, \x <_R \y \Rightarrow \x <_D \y$. Note that agreement is not necessarily a symmetric property.
\end{definition}

\begin{proposition}
\label{prop:exist_agree}
    For every preorder $R$, there exists a total order $T$ that agrees with $R$.
\end{proposition}

\begin{proof}
    We proceed in two stages by first showing that there exists a partial order that agrees with $R$ and then invoking the Szpilrajn extension theorem~\cite{szpilrajn1930extension} to argue that there exists a total order that agrees with the partial order. 

    Every preorder $R$ is a partial order over equivalence classes. Construct partial order $D$ by maintaining all non-equivalence relations of $R$ and arbitrarily assigning a total order between the elements within each of the equivalence classes of $R$ (that such a total order exists follows from the well-ordering principle). Now take any $\x, \y \in \Omega$ such that $\x <_R \y$. Since $\x \not \sim_R \y$, it's clear by construction that $\x <_D \y$ as well. Thus, partial order $D$ agrees with $R$.  
    
    The Szpilrajn extension theorem~\cite{szpilrajn1930extension} establishes that there must also exist a total order $T$ such that $\forall \x, \y \in \Omega, \x \leq_D \y \Rightarrow \x \leq_T \y$. Fix $\x, \y \in \Omega$ such that $\x <_D \y$ and let $T$ be the total order guaranteed by the theorem. We must then have that $\x \leq_T \y$. And of course $\x \neq \y$ by Definition~\ref{def:partial_order}. Thus, by Definition~\ref{def:total_order}, $\x <_T \y$ as well. Therefore, total order $T$ agrees with $D$.
\end{proof}

\begin{proposition}
\label{prob:agree_implies_upper_subset}
    If total order $T$ agrees with total preorder $R$, then $\forall \x \in \Omega$, $\Omega(\x, T) \subseteq \Omega(\x, R)$.
\end{proposition}

\begin{proof}  
    Let $\x \in \Omega$ be arbitrary and note that by  Definition~\ref{def:agreement}, if total order $T$ agrees with total preorder $R$, then for each $\y \in \Omega$
    \begin{equation}
    \label{eq:contra_agreement}
        \neg(\y <_T \x) \Rightarrow \neg(\y <_R \x).
    \end{equation}
    It follows that
    \begin{equation}
        \begin{array}{rcll}
            \Omega(\x, T) &=& \{\y \in \Omega: \x \leq_T \y\} & \text{Definition~\ref{def:pre_upper}}\\
            &=& \{\y \in \Omega: \neg(\y <_T \x)\} & \text{Proposition~\ref{prop:less_implies_not_greater_equal}} \\
            &\subseteq& \{\y \in \Omega: \neg(\y <_R \x)\} & \text{Implication~\ref{eq:contra_agreement}} \\
            &=& \{\y \in \Omega: \x \lesssim_R \y)\} & \text{Proposition~\ref{prop:less_implies_not_greater_equal}} \\
            &=& \Omega(\x, R) & \text{Definition~\ref{def:pre_upper}}\\
        \end{array}.
    \end{equation}
\end{proof}

\begin{proposition}
\label{prop:agree_implies_upper_equal}
    For every total order $T$ agreeing with total preorder $R$ and any $\x \in \Omega$, there exists a $\y \in \Omega$ such that $\x \sim_R \y$ and $\Omega(\x, R) = \Omega(\y, T)$.
\end{proposition}

\begin{proof}
    Fix $\x \in \Omega$ and take $\y$ such that $\x \sim_R \y$ and $\forall \z \in \Omega : \x \sim_R \z$, $\y \leq_T \z$, i.e. $\y$ is the smallest element according to $T$ that is equivalent to $\x$ according to $R$. Such a $\y$ must exist by the well-ordering principle. By Definition~\ref{def:pre_upper}, $\Omega(\x, R) = \Omega(\y, R)$ so that, invoking Proposition~\ref{prob:agree_implies_upper_subset}, it suffices to show that $\Omega(\y, R) \subseteq \Omega(\y, T)$.
    
    We can partition $\Omega(\y, R)$ such that $\Omega(\y, R) = S_1 \cup S_2$ where $S_1 = \{\z: \y \sim_R \z\}$ and $S_2 = \{\z: \y <_R \z\}$. Now take any $\z \in \Omega(\y, R)$. If $\z = \y$, then it's clear by Definition~\ref{def:pre_upper} that $\z \in \Omega(\y, T)$. If $\z \in S_1 \setminus \{\y\}$, then $\y <_T \z$ by our choice of $\y$, so that, again by Definition~\ref{def:pre_upper}, $\z \in \Omega(\y, T)$. Finally, suppose that $\z \in S_2$. Since $T$ agress with $R$, we have by Definition~\ref{def:agreement} that $\y <_T \z$. Thus, by Definition~\ref{def:pre_upper}, $\z \in \Omega(\y, T)$, and therefore $\Omega(\y, R) \subseteq \Omega(\y, T)$.
\end{proof}

\begin{definition}
\label{def:total_order_consistent}
    A bound $B$ is consistent with total order $T$ if $\forall \x,\y \in \Omega$, 
    \begin{equation}
        \x \leq_T \y \Rightarrow B(\x) \leq B(\y).
    \end{equation}
\end{definition}

\begin{definition}
\label{def:preorder_consistent}
    Bound $B$ is consistent with preorder $R$ if it is consistent with every total order that agrees with $R$.
\end{definition}

\begin{lemma}
\label{lem:pre_consistent}
    If bound $B$ is consistent with preorder $R$ then for any $\x, \y \in \Omega$, $B(\x) \leq B(\y)$ whenever $\x <_R \y$ and $B(\x) = B(\y)$ whenever $\x \sim_R \y$. The converse also holds if $R$ is a total preorder.
\end{lemma}
\begin{proof}
    Suppose $B$ is consistent with preorder $R$ and take first any $\x, \y \in \Omega$ where $\x <_R \y$. Definition~\ref{def:preorder_consistent} dictates that $B$ must be consistent with every total order $T$ agreeing with $R$. And Definition~\ref{def:agreement} requires that we must also have $\x <_T \y$, which in turn implies, by Definition~\ref{def:total_order_consistent}, that $B(\x) \leq B(\y)$.
    Next, take $\x, \y \in \Omega$ such that $\x \sim_R \y$. Let $T$ be any total order agreeing with $R$ such that $\x <_T \y$, and let $T'$ be the total order identical to $T$ except that it reverses the order of $\x$ and $\y$. Since $B$ is consistent with $R$, Definition~\ref{def:preorder_consistent} ensures that $B$ is consistent with both $T$ and $T'$. This implies that $B(\x) \leq B(\y)$ and $B(\y) \leq B(\x)$ so that $B(\x) = B(\y)$. 

    Now suppose that $R$ is a total preorder, $B(\x) = B(\y)$ whenever $\x \sim_R \y$, and $B(\x) \leq B(\y)$ whenever $\x <_R \y$ for preorder $R$. Let $T$ be any total order agreeing with $R$. For each $\x, \y \in \Omega$ we have by Proposition~\ref{prop:less_implies_not_greater_equal} and taking the contrapositive of Definition~\ref{def:agreement} that
    \begin{equation}
    \label{eq:total_implies_pre}
        \x \leq_T \y \Rightarrow \neg(\y <_T \x) \Rightarrow \neg(\y <_R \x) \Rightarrow \x \lesssim_R \y.
    \end{equation}
    If $\x <_R \y$, then by assumption $B(\x) \leq B(\y)$. And if $\x \sim_R \y$, we again have by assumption that $B(\x) = B(\y) \Rightarrow B(\x) \leq B(\y)$. Thus, in any case, $\x \lesssim_R \y \Rightarrow B(\x) \leq B(\y)$. It follows then by Implication~\ref{eq:total_implies_pre} that $\x \leq_T \y \Rightarrow B(\x) \leq B(\y)$. Therefore, by Definition~\ref{def:total_order_consistent}, it must be the case that $B$ is consistent with $T$. Since $T$ was an arbitrary total order agreeing with $R$, we have, by Definition~\ref{def:preorder_consistent}, that $B$ is also consistent with $R$.
\end{proof}

\begin{definition}
\label{def:cond_opt_pre}
    Bound $B$ consistent with preorder $R$ is \emph{conditionally optimal} with respect to $R$ if for any other bound $B'$ also consistent with $R$ we have $\forall \x \in \Omega, B'(\x) \leq B(\x)$, and $\exists \y \in \Omega: B'(\y) < B(\y)$. 
\end{definition}

\begin{definition}
\label{def:pessimal_bound}
    Let $\x \in \Omega$ and $\alpha \in [0, 1)$ be arbitrary. The \emph{pessimal bound} with respect to preorder $R$ is given by
    \begin{equation}
    \label{eq:cond_opt}
        B_R^*(\x) = \min \{E[F]: F \in \mathcal{F}(\Omega(\bm x, R), \alpha)\}.
    \end{equation}
\end{definition}

\begin{theorem}[Learned-Miller~\cite{learnedmiller2025admissibilityboundsmeandiscrete}]
\label{thm:conditional_optimality}
    Let $\x \in \Omega$, $\alpha \in [0, 1)$, and total order $T$ be arbitrary. When $\mathcal{F}(\Omega(\bm x, T), \alpha)$ is non-empty, the pessimal bound $B_T^*(\x)$ is conditionally optimal with respect to $T$.
\end{theorem}

For the remainder of this document we will assume that for any preorder $R$ under consideration, $\alpha$ is chosen so that $\mathcal{F}(\Omega(\bm x, R), \alpha)$ is indeed non-empty.

\begin{lemma}
\label{lem:pessimal_consistent}
    For any total preorder $R$, $B_R^*$ is consistent with $R$.
\end{lemma}
\begin{proof}
    According to Lemma~\ref{lem:pre_consistent}, it will suffice to show that $B_R^*(\x) \leq B_R^*(\y)$ whenever $\x <_R \y$ and $B_R^*(\x) = B_R^*(\y)$ whenever $\x \sim_R \y$. We can see that the latter is true by noting that Definition~\ref{def:pre_upper} ensures $\Omega(\x, R) = \Omega(\y, R)$ whenever $\x \sim_R \y$ so that $B_R^*(\x) = B_R^*(\y)$ by Definition~\ref{def:pessimal_bound}. To show the former, note that by Definition~\ref{def:pre_upper}, $\x \lesssim_R \y \Rightarrow \Omega(\y, R) \subseteq \Omega(\x, R)$, which by Definition~\ref{def:pessimal_bound} implies that $B_R^*(\x) \leq B_R^*(\y)$.
\end{proof}

\begin{theorem}
\label{thm:preorder_cond_opt}
    For every $\alpha \in [0, 1)$ and total preorder $R$, $B_R^*$ is the conditionally optimal bound for $R$ provided that $\mathcal{F}(\Omega(\x, R), \alpha)$ is non-empty. 
\end{theorem}
\begin{proof}
    We begin by noting that, according to Lemma~\ref{lem:pessimal_consistent}, $B_R^*$ is consistent with $R$.
    Thus, by Definition~\ref{def:cond_opt_pre}, it suffices to show that for any valid bound $B_R$, consistent with $R$, and $\forall \x \in \Omega, B^*_R(\x) \geq B_R(\x)$. This would imply that $B_R$ is either strictly worse than $B_R^*$ for some $\x \in \Omega$, or that $B_R = B_R^*$. To that end, suppose to the contrary that there exists some valid $B_R$, also consistent with $R$, such that $\exists \x \in \Omega: B_R^*(\x) < B_R(\x)$. Now let $T$ be any total order agreeing with $R$ and take $\y \in \Omega$ such that $\y \sim_R \x$ and $\Omega(\y, T) = \Omega(\y, R) = \Omega(\x, R)$. We are guaranteed that such a $T$ exists by Propositions~\ref{prop:exist_agree}, and we are guaranteed that such a $\y$ exists by Proposition~\ref{prop:agree_implies_upper_equal}. 
    Since $T$ agrees with $R$ we know by Definition~\ref{def:preorder_consistent} that $B_R$ is consistent with $T$. Notice also that Theorem~\ref{thm:conditional_optimality} ensures that $B_T^*$ is conditionally optimal for $T$.
    
    Now since $\x \sim_R \y$ and $B_R^*(\x) < B_R(\x)$, it's clear by Lemma~\ref{lem:pre_consistent} that 
    \begin{equation}
        B_R^*(\y) = B_R^*(\x) < B_R(\x) = B_R(\y).
    \end{equation}
    And because $\Omega(\y, T) = \Omega(\y, R)$ we also have by Definition~\ref{def:pessimal_bound} that $ B_T^*(\y) = B_R^*(\y)$. So it must be the case that $B_T^*(\y) < B_R(\y)$ as well, which means that there is a bound consistent with $T$ (the bound $B_R$) that for sample $\y$ improves on the bound given by $B_T^*$. But this contradicts the fact that $B_T^*$ is conditionally optimal for order $T$. Thus, it cannot be the case that $B_R^*(\x) < B_R(\x)$.
\end{proof}

\section{Pointwise-Optimality}

The theory of Learned-Miller~\cite{learnedmiller2025admissibilityboundsmeandiscrete} establishes that there can exist no bound that is optimal for \emph{every} sample in $\Omega$. Nevertheless, it remains possible to derive the highest possible bound value that can be assigned to any given sample by any valid bound. 

\begin{definition}
\label{def:point_opt}
    The \emph{optimal bound} for sample $\x \in \Omega$, denoted $B^*(\x)$, is the highest bound value assigned to $\x$ by any valid bound.
\end{definition}

\begin{lemma}
\label{lem:point_opt}
    For any $\x \in \Omega$,
    \begin{equation}
        B^*(\x) = \min \{E[F]: F \in \mathcal{F}(\{\x\}, \alpha)\}.
    \end{equation}
\end{lemma}
\begin{proof}
    From Definition~\ref{def:point_opt} and Theorem~\ref{thm:conditional_optimality}, it will suffice to find the highest bound value for $\x$ among all pessimal bounds (see Definition~\ref{def:pessimal_bound}). These bounds amount to finding the distribution $F^*$ achieving minimum mean among all distributions in a set of the form $\{F: F \in \mathcal{F}(\Omega(\x, T), \alpha)\}$ for some total order $T$. The conclusion follows by noticing that each of these sets contains the set $\{F: F \in \mathcal{F}(\{\x\}, \alpha)\}$, which corresponds to any total order that places $\x$ last in its ordering.
\end{proof}

\begin{theorem}
The optimal bound for homogeneous $\bm S_i \in \Omega$ is given by
    \begin{equation}
        B^*(\bm S_i) = S_{\min} (1 - \sqrt[n]{\alpha}) + S_i \sqrt[n]{\alpha}.
    \end{equation}
\end{theorem}
\begin{proof}
    From Lemma~\ref{lem:point_opt}, we know that $B^*(\bm S_i) = \min \{E[F]: F \in \mathcal{F}(\{\bm S_i\}, \alpha)\}$. Let $p_0 = \sqrt[n]{\alpha}$. Every distribution $G_p \in \mathcal{G}(\{\bm S_i\}, \alpha)$ places mass $p = p_0 + \epsilon$ on support point $S_i$, for some $\epsilon > 0$.
    Thus, according to Lemma~\ref{lem:closure_mean}, $B^*(\bm S_i)$ will be equal to the limiting mean of some sequence $\{G_{p_0 + \epsilon}\}_{\epsilon \downarrow 0}$.
    Let $H_p \in \mathcal{G}(\{\bm S_i\}, \alpha)$ be the distribution that places mass $p$ at support point $S_i$, and places the remaining mass at $S_{\min}$. It is clear that 1) $G_p \in \mathcal{G}(\{\bm S_i\}, \alpha) \Rightarrow H_p \in \mathcal{G}(\{\bm S_i\}, \alpha)$ 2) every sequence $\{G_{p_0 + \epsilon}\}_{\epsilon \downarrow 0}$ is bounded below by the corresponding sequence $\{H_{p_0 + \epsilon}\}_{\epsilon \downarrow 0}$, and 3) $E[H_p] \leq E[G_p]$. Therefore, 
    \begin{equation}
        B^*(\bm S_i) = \lim_{\epsilon \rightarrow 0} E[H_{p_0 + \epsilon}] = S_{\min} (1 - \sqrt[n]{\alpha}) + S_i \sqrt[n]{\alpha}.
    \end{equation}
\end{proof}

\section{Monotone Orders}

In this section we show that the low and high lexicographic orders are extremal in the sense that the conditionally-optimal bound for any monotone order evaluated at a given sample will fall between the values given for \emph{nearby} samples by bounds conditionally-optimal for the low and high lexicographic orders.

\begin{proposition}
\label{prop:lexi_mono}
    The low and high lexicographic orders are monotone.
\end{proposition}
\begin{proof}
    We prove the proposition for $T_\ell$; the proof for $T_h$ is similar. Suppose that for any $\bm x, \bm y \in \Omega$, $\bm x \leq \bm y$ but $\neg (\bm x \leq_{T_\ell} \bm y)$. Since $\neg (\bm x \leq_{T_\ell} \bm y)$ we know by definition of $T_\ell$ that $L_1(\bm x, \bm y) = 0$, which implies that $\exists i \in [n]: x_{(i)} > y_{(i)}$, i.e. $\neg (\bm x \leq \bm y)$. But this contradicts the assumption that $\bm x \leq \bm y$. Thus, we must have $\bm x \leq \bm y \Rightarrow \bm x \leq_T \bm y$.
\end{proof}

\begin{proposition}
\label{prop:lexi_low_equiv}
    Let $\bm y \in \Omega$ be arbitrary and $\bm S_i \in \Omega$ be any homogeneous sample. For the low lexicographic order $T_\ell$ we have
    \begin{equation}
        \bm S_i \leq_{T_\ell} \bm y \Leftrightarrow \bm S_i \leq \bm y.
    \end{equation}
\end{proposition}
\begin{proof}
    Since $T_\ell$ is monotone, it suffices to show that $\bm S_i \leq_{T_\ell} \bm y \Rightarrow \bm S_i \leq \bm y$. Suppose to the contrary that $\bm S_i \leq_{T_\ell} \bm y$ but $\neg(\bm S_i \leq \bm y)$. Since $\neg(\bm S_i \leq \bm y)$ we know that $\exists j \in [n]: S_i > y_{(j)}$, which in turn implies that $\forall k \leq j, S_i > y_{(k)}$. But because $\bm S_i \leq_{T_\ell} \bm y$, we have that $L_1(\bm S_i, \bm y) = 1$. This implies in particular that $S_i \leq y_{(1)}$, which leads to a contradiction. Thus, it must be the case that $\bm S_i \leq_{T_\ell} \bm y \Rightarrow \bm S_i \leq \bm y$.
\end{proof}

\begin{proposition}
\label{prop:lexi_high_equiv}
    Let $\bm y \in \Omega$ be arbitrary and $\bm S_i \in \Omega$ be any homogeneous sample. For the high lexicographic order $T_h$ we have
    \begin{equation}
        \bm S_i \leq_{T_h} \bm y \Leftrightarrow \bm S_i = \bm y \vee S_i < y_{(n)}.
    \end{equation}
\end{proposition}
\begin{proof}
    The equivalence is clearly true when $\bm S_i = \bm y$, so we proceed under the assumption that $\bm S_i \neq \bm y$.
    We begin by showing that $\bm S_i \leq_{T_h} \bm y \Rightarrow S_i < y_{(n)}$. Suppose to the contrary that 
    $\bm S_i \leq_{T_h} \bm y$ but $y_{(n)} \leq S_i$. Since $\bm S_i$ is homogeneous it must also be the case that $\bm y \leq \bm S_i$. Combined with the assumption that $\bm S_i \neq \bm y$, this implies that $\bm y < \bm S_i$. But since $T_h$ is monotone we also have that $\bm y < \bm S_i \Rightarrow \bm y <_{T_h} \bm S_i$, which contradicts the original assumption that $\bm S_i \leq_{T_h} \bm y$. Thus, $\bm S_i \leq_{T_h} \bm y \Rightarrow S_i < y_{(n)}$.

    Next we show that $S_i < y_{(n)} \Rightarrow \bm S_i \leq_{T_h} \bm y$. From $S_i < y_{(n)}$ it is immediately clear that $H_n(\bm S_i, \bm y) = 1$, which implies by definition that $\bm S_i \leq_{T_h} \bm y$. Thus, $S_i < y_{(n)} \Rightarrow \bm S_i \leq_{T_h} \bm y$.
\end{proof}

\begin{theorem}
\label{thm:lexi_sandwich}
    For any monotone total order $T$ and homogeneous sample $\bm S_i \in \Omega$ we have
    \begin{equation}
        \Omega(\bm S_i, T_\ell) \subseteq \Omega(\bm S_i, T) \subseteq \Omega(\bm S_i, T_h).
    \end{equation}
\end{theorem}
\begin{proof}
    We first show that $\Omega(\bm S_i, T_\ell) \subseteq \Omega(\bm S_i, T)$. Assume to the contrary that $\exists \bm y \in \Omega$ such that $\bm y \in \Omega(\bm S_i, T_\ell)$ and $\bm y \not \in \Omega(\bm S_i, T)$, the latter of which implies $\neg (\bm S_i \leq_T \bm y)$. Since $T$ is monotone we have $\neg (\bm S_i \leq_T \bm y) \Rightarrow \neg (\bm S_i \leq \bm y)$. Therefore, it must be the case that $\neg (\bm S_i \leq \bm y)$. On the other hand, since $\bm y \in \Omega(\bm S_i, T_\ell)$ we know by Proposition~\ref{prop:lexi_low_equiv} that $\bm S_i \leq \bm y$, which leads to a contradiction. Thus, $\Omega(\bm S_i, T_\ell) \subseteq \Omega(\bm S_i, T)$.

    Next, we show that $\Omega(\bm S_i, T) \subseteq \Omega(\bm S_i, T_h)$. Suppose to the contrary that $\exists \bm y \in \Omega: \bm y \in \Omega(\bm S_i, T)$ and $\bm y \not \in \Omega(\bm S_i, T_h)$. 
    Since $\bm y \not \in \Omega(\bm S_i, T_h)$, we know by Proposition~\ref{prop:lexi_high_equiv} that $\bm y \neq \bm S_i$ and $y_{(n)} \leq S_i$, which together imply that $\bm y < \bm S_i$. 
    And since $T$ is monotone, $\bm y < \bm S_i$ implies that $\bm y <_T \bm S_i$. On the other hand, $\bm y \in \Omega(\bm S_i, T)$ implies that $\bm S_i \leq_T \bm y$, which leads to a contradiction. Thus, we have shown that $\Omega(\bm S_i, T) \subseteq \Omega(\bm S_i, T_h)$.
\end{proof}

\begin{theorem}
\label{thm:sandwich}
    Let $T$ be any monotone total order, $B_T^*$ the conditionally optimal bound for $T$, and $\bm x \in \Omega$ an arbitrary sample. If $S_i \in S$ is such that $\bm S_i \leq_T \bm x \leq_T \bm S_{i+1}$, then
    \begin{equation}
        B_{T_h}^*(\bm S_i) \leq B_T^*(\bm x) \leq B_{T_\ell}^*(\bm S_{i+1}).
    \end{equation}
\end{theorem}
\begin{proof}
    We begin by stating some facts
    \begin{enumerate}
        \item For upper sets $\Omega_1 \subseteq \Omega_2 \subseteq \Omega$ we have $\mathcal{F}(\Omega_1) \subseteq \mathcal{F}(\Omega_2)$: Take any $F \in \mathcal{F}(\Omega_1)$. By construction, $P_F[\Omega_1] > \alpha$. But since $\Omega_1 \subseteq \Omega_2$, it must also be the case that $P_F[\Omega_2] > \alpha$, which implies that $F \in \mathcal{F}(\Omega_2)$ as well.
        \item Since $B_T^*$ is consistent with $T$, and $\bm S_i \leq_T \bm x \leq_T \bm S_{i+1}$, we know that $B_T^*(\bm S_i) \leq B_T^*(\bm x) \leq B_T^*(\bm S_{i+1})$.
        \item For any subsets of distributions $\mathcal{F}_1 \subseteq \mathcal{F}_2 \subseteq \mathcal{F}$ we have that that $\min\{E[F]: F \in \mathcal{F}_2\} \leq \min\{E[F]: F \in \mathcal{F}_1\}$.
        \item For any order $T'$ we know by Theorem~\ref{thm:conditional_optimality} that $B_{T'}^*(\bm S_i) = \min\{E[F]: F \in \mathcal{F}(\Omega(\bm S_i, T'), \alpha)\}$. 
    \end{enumerate}
    We now show that  $B_T^*(\bm x) \leq B_{T_\ell}^*(\bm S_{i+1})$. By Theorem~\ref{thm:lexi_sandwich} we have that $\Omega(\bm S_i, T_\ell) \subseteq \Omega(\bm S_i, T)$. Thus, according to (1), $\mathcal{F}(\Omega(\bm S_i, T_\ell), \alpha) \subseteq \mathcal{F}(\Omega(\bm S_i, T), \alpha)$. Finally, from (2-4) we conclude that
    \begin{equation}
    \begin{array}{rcl}
        B_T^*(\bm x) &\leq& B_T^*(\bm S_{i+1}) \\
        &=&\min\{E[F]: F \in \mathcal{F}(\Omega(\bm S_{i+1}, T), \alpha)\} \\
        &\leq& \min\{E[F]: F \in \mathcal{F}(\Omega(\bm S_{i+1}, T_\ell), \alpha)\} \\
        &=& B_{T_\ell}^*(\bm S_{i+1}) \\
    \end{array}.
    \end{equation}
     We take a similar approach to show that $B_{T_h}^*(\bm S_i) \leq B_T^*(\bm x)$. Again by Theorem~\ref{thm:lexi_sandwich} we know that $\Omega(\bm S_i, T) \subseteq \Omega(\bm S_i, T_h)$, which combined with (1) implies that $\mathcal{F}(\Omega(\bm S_i, T), \alpha) \subseteq \mathcal{F}(\Omega(\bm S_i, T_h), \alpha)$. Invoking (2-4) gives
    \begin{equation}
    \begin{array}{rcl}
        B_{T_h}^*(\bm S_i) &=& \min\{E[F]: F \in \mathcal{F}(\Omega(\bm S_i, T_h), \alpha)\} \\
        &\leq& \min\{E[F]: F \in \mathcal{F}(\Omega(\bm S_i, T), \alpha)\} \\
        &=& B_T^*(\bm S_i) \\
        &\leq& B_T^*(\bm x) \\
    \end{array}.
    \end{equation}
\end{proof}

\begin{corollary}
\label{cor:lexi_low_opt}
    For any homogeneous sample $S_i \in S$, $B_{T_h}^*(\bm S_i)$ is the weakest (lowest) and $B_{T_\ell}^*(\bm S_i)$ the strongest (highest) among all bounds conditionally optimal with respect to a monotone total order.
\end{corollary}

\begin{proof}
    Let $T$ be any monotone order and $B_T^*$ the conditionally optimal bound for that order. According to Theorem~\ref{thm:sandwich} we have simultaneously that
    \begin{equation}
        B_{T_h}^*(\bm S_{i-1}) \leq B_T^*(\bm S_i) \leq B_{T_\ell}^*(\bm S_i),
    \end{equation}
    and
    \begin{equation}
        B_{T_h}^*(\bm S_i) \leq B_T^*(\bm S_i) \leq B_{T_\ell}^*(\bm S_{i+1}),
    \end{equation}
    which implies that
    \begin{equation}
        B_{T_h}^*(\bm S_i) \leq B_T^*(\bm S_i) \leq B_{T_\ell}^*(\bm S_i).
    \end{equation}
\end{proof}

\subsection{Calculating bounds for homogeneous samples}

\begin{lemma}
\label{lem:lexi_low_admis}
    Bound $B_{T_\ell}^*$ is conditionally optimal at the $1-\alpha$ level with respect to $\mathcal{F}$ only if $\forall S_i \in S$, 
    \begin{equation}
        B_{T_\ell}^*(\bm S_i) = S_{\min} (1 - \sqrt[n]{\alpha}) + S_i \sqrt[n]{\alpha}.
    \end{equation}
\end{lemma}
\begin{proof}
    Fix $\bm S_i$. Learned-Miller~\cite{learnedmiller2025admissibilityboundsmeandiscrete} Lemma 2.2 establishes that the result is true for $S_i = S_{\min}$, so we assume that $S_i \neq S_{\min}$.
    By Theorem~\ref{thm:conditional_optimality} it will suffice to show that $\min\{E[F]: F \in \mathcal{F}(\Omega(\bm S_i, T_\ell), \alpha)\} = \mu$ where $\mu = S_{\min} (1 - \sqrt[n]{\alpha}) + S_i \sqrt[n]{\alpha}$.
    
    On one hand we have by Proposition~\ref{prop:lexi_low_equiv} that for all $\bm y \in \Omega$, $\bm S_i \leq_{T_\ell} \bm y \Rightarrow \bm S_i \leq \bm y$, i.e. all samples in $\Omega(\bm S_i, T_\ell)$ comprise support values greater than or equal to $S_i$.  
    Therefore, for every $G \in \mathcal{G}(\Omega(\bm S_i, T_\ell), \alpha), E[G] > \mu$ since $G$ must assign mass exceeding $\alpha$ to samples in $\Omega(\bm S_i, T_\ell)$, which implies that the mass assigned to support values $S_i$ or higher must exceed $\sqrt[n]{\alpha}$. Thus, every sequence of distributions from $\mathcal{G}(\Omega(\bm S_i, T_\ell), \alpha)$ must converge in mean to some value no less than $\mu$ so that, by Proposition~\ref{lem:closure_mean}, $\min\{E[F]: F \in \mathcal{F}(\Omega(\bm S_i, T_\ell), \alpha)\} \geq \mu$. 
    On the other hand, there exists a $j \in \mathbb{N}$ for which we can construct a sequence of distributions $G_j, G_{j+1} \ldots \in \mathcal{G}(\Omega(\bm S_i, T_\ell), \alpha)$ such that for each $i \geq j$, $G_i(S_{\min}) = (1 - \sqrt[n]{\alpha}) - 1/i$ and $G_i(S_i) = \sqrt[n]{\alpha} + 1/i$. In this case we have that $\forall \delta > 0, \exists N \in \mathbb{N}: \forall n > N, E[G_n] - \mu < \delta$, which means that the sequence converges in mean to $\mu$. Again by Proposition~\ref{lem:closure_mean}, this implies that there exists a distribution in $\mathcal{F}(\Omega(\bm S_i, T_\ell), \alpha)$ with mean $\mu$ so that $\min\{E[F]: F \in \mathcal{F}(\Omega(\bm S_i, T_\ell), \alpha)\} \leq \mu$.
\end{proof}

\begin{lemma}
\label{lem:lexi_high_admis}  
    Bound $B_{T_h}^*$ is conditionally optimal at the $1-\alpha$ level with respect to $\mathcal{F}$ only if $\forall S_i \in S$, $i \neq 0$, 
    \begin{equation}
        \label{eq:high_lexi_int}
        S_{\min} \sqrt[n]{1-\alpha} + S_i (1 - \sqrt[n]{1-\alpha}) \leq B_{T_h}^*(\bm S_i) \leq S_{\min} \sqrt[n]{1-\alpha} + S_{i+1} (1 - \sqrt[n]{1-\alpha}).
    \end{equation}
\end{lemma}
\begin{proof}
    From Learned-Miller~\cite{learnedmiller2025admissibilityboundsmeandiscrete} Lemma 2.2, the result is true for $S_i = S_{\min}$, so we fix $\bm S_i \neq \bm S_{\min}$. By Theorem~\ref{thm:conditional_optimality} we need only show that $\min\{E[F]: F \in \mathcal{F}(\Omega(\bm S_i, T_h), \alpha)\} = \mu$ falls within the interval defined by Inequality~\ref{eq:high_lexi_int}. To that end, by Proposition~\ref{prop:lexi_low_equiv}, it will suffice to show the following two properties. On one hand, $\forall H \in \mathcal{G}(\Omega(\bm S_i, T_h), \alpha)\}$, $E[H] > S_{\min} \sqrt[n]{1-\alpha} + S_i (1 - \sqrt[n]{1-\alpha})$. And on the other hand, there exists a sequence of distributions in $\mathcal{G}(\Omega(\bm S_i, T_h), \alpha)$ whose limiting mean is bounded above by $S_{\min} \sqrt[n]{1-\alpha} + S_{i+1} (1 - \sqrt[n]{1-\alpha})$. 

    Fix $S_i \in S$, $i > 0$, and let $H \in \mathcal{G}(\Omega(\bm S_i, T_h), \alpha)$ be arbitrary. Define transformation $\phi: S \rightarrow S$ such that for each $S_k \in S$, $\phi(S_k) = S_{\min}$ if $S_k < S_i$ and $\phi(S_k) = S_{i+1}$ if $S_k \geq S_i$. Furthermore, define $G_H \in \mathcal{F}$ to be the distribution that, for each $S_k \in S$, transfers all mass from $S_k$ to $\phi(S_k)$. Note that $\forall A \in \Omega, P_{G_H}[\phi(A)] \geq P_H[A]$. Let $A = \Omega(\bm S_i, T_h)$. We next show that $P_{G_H}[A] > \alpha$ so that $G_H \in \mathcal{G}(\Omega(\bm S_i, T_h), \alpha)$. To that end, it will suffice to show that $\phi(A) \subseteq A$ so that $P_{G_H}[A] \geq P_{G_H}[\phi(A)] \geq P_H[A] > \alpha$. Take any $\bm x \in A$. If $\bm x = \bm S_i$, then $\phi(\bm x) = \bm S_{i+1} \in A$. Otherwise, by Proposition~\ref{prop:lexi_high_equiv}, $x_{(n)} \geq S_i$ and so $\phi(x_{(n)}) = S_{i+1} > S_i$, which also implies that $\phi(\bm x) \in A$. Thus, $\phi(A) \subseteq A$.

    Now define $G_H'$ as the distribution that results from $G_H$ after transferring all mass in $G_H$ from $S_{i+1}$ to $S_i$. It's clear by construction that $H$ stochastically dominates $G_H'$ so that $E[G_H'] \leq E[H]$. Notice that both $G_H$ and $G_H'$ can be reparameterized by $p$, where fraction $p$ of the mass in the distribution is placed at $S_{\min}$ and the rest is placed at the other atom in their support. 
    We say that $G_p$ and $G_p'$ are \emph{achievable} if $\exists H \in \mathcal{G}(\Omega(\bm S_i, T_h), \alpha)$ such that $G_H = G_p$.
    Under this parameterization we have $E[G_p] = p S_{\min} + (1-p)S_{i+1}$ and $E[G_p'] = p S_{\min} + (1-p)S_i$.

    Finally, define $p^* = \sqrt[n]{1-\alpha}$. Notice that the only sample from $\phi(\Omega)$ not in $\Omega(\bm S_i, T_h)$ is $\bm S_{\min}$. Thus, the pair $G_p'$ and $G_p$ is achievable iff we have that $p < p^*$. We have argued that for every $H \in \mathcal{G}(\Omega(\bm S_i, T_h), \alpha)$ there exists an achievable $G_p'$ such that $E[G_p'] < E[H]$. Therefore, $\forall H \in \mathcal{G}(\Omega(\bm S_i, T_h), \alpha), E[G_{p^*}'] < E[H]$. Now define the sequence of distributions $G_{p^* - 1/j}$ for each $j > \lceil 1/p^* \rceil$. By construction, each element in the sequence is achievable and the sequence converges to $G_{p^*}$. We have also argued that corresponding to each $G_{p^* - 1/j}$ there exists $H \in \mathcal{G}(\Omega(\bm S_i, T_h), \alpha)$ such that $E[H] \leq E[G_{p^*-1/j}]$. Therefore, there exists a sequence of distributions in $\mathcal{G}(\Omega(\bm S_i, T_h), \alpha)$ whose means are bounded from above by the means of the sequence of distributions $G_{p^* - 1/j}$. 
\end{proof}

\section{Refinements of $\mathcal{F}$}

\begin{definition}
    For any $G, G' \in \mathcal{F}$ and $C \subseteq S$, we say that $G$ and $G'$ agree \emph{pointwise} on $C$ if $\forall x \in C, P_G[X = x] = P_{G'}[X = x]$ and \emph{cumulatively} on $C$ if $\forall x \in C, P_G[\X \leq \x] = P_{G'}[\X \leq \x]$.
\end{definition}

\begin{definition}
    For arbitrary $C \in S$ define the \emph{refinement} $\mathcal{F}_C = \{F \in \mathcal{F}: \forall S_i \in S \setminus C, P_F[S_i] = 0\}$. For arbitrary $\Omega' \subseteq \Omega$, we further define refinements
    \begin{equation}
        \mathcal{G}_C(\Omega', \alpha) = \mathcal{F}_C \cap \mathcal{G}(\Omega', \alpha),
    \end{equation}
    and
    \begin{equation}
        \mathcal{F}_C(\Omega', \alpha) = \texttt{cl}(\mathcal{G}_C(\Omega', \alpha)).
    \end{equation}
\end{definition} 

\begin{definition}
\label{def:augmentation}
    For each $C \subseteq S$, we denote by $C^+$ the augmentation 
    \begin{equation}
        C^+ = C \cup \{S_{\min}\} \cup \{S_{i+1}: S_i \in C, S_i \neq S_{\max}\}.
    \end{equation}
\end{definition}

\begin{lemma}
\label{lem:support_refinement}
    Fix $\x \in \Omega$, $\alpha \in [0,1)$, and total preorder $R$. Suppose that there exists a set $C \subseteq S$ such that for any $G, H \in \mathcal{F}$ agreeing cumulatively and pointwise on $C$, $P_G[\Omega(\x, R)] = P_H[\Omega(\x, R)]$. Then we have that
    \begin{equation}
        B_R^*(\x) = \min\{E[F]: F \in \mathcal{F}_{C^+}(\Omega(\bm x, R), \alpha)\}.
    \end{equation}
\end{lemma}
\begin{proof}
    From Theorem~\ref{thm:preorder_cond_opt} and Lemma~\ref{lem:closure_mean} we have: 
    \begin{enumerate}
        \item For every sequence $\{G_i\}_{i>0}$, $G_i \in \mathcal{G}(\Omega(\x, R), \alpha)$, $\lim_{i \rightarrow \infty} E[G_i] \geq B_R^*(\x)$.
        \item There exists a sequence of distributions $\{G_i^*\}_{i>0}$, $G_i^* \in \mathcal{G}(\Omega(\x, R), \alpha)$, such that $\lim_{i \rightarrow \infty} E[G_i^*] = B_R^*(\x)$ 
    \end{enumerate}
    We will first show that for every $G \in \mathcal{G}(\Omega(\x, R), \alpha)$, there exists an $H \in \mathcal{G}_{C^+}(\Omega(\x, R), \alpha)$ such that $E[H] \leq E[G]$. 
    This establishes that, for the sequence $\{G_i^*\}_{i>0}$ in particular, there exists corresponding sequence $\{H_i^*\}_{i>0}$, $H_i^* \in \mathcal{G}_{C^+}(\Omega(\x, R), \alpha)$, such that $\lim_{i \rightarrow \infty} E[H_i^*] \leq \lim_{i \rightarrow \infty} E[G_i^*]$, which would imply by Property 2 that $\lim_{i \rightarrow \infty} E[H_i^*] \leq B_R^*(\x)$. On the other hand, since $\mathcal{G}_{C^+}(\Omega(\x, R), \alpha) \subseteq \mathcal{G}(\Omega(\x, R), \alpha)$, Property 1 establishes that $\lim_{i \rightarrow \infty} E[H_i^*] \geq B_R^*(\x)$. Thus, it will follow that the distribution with lowest mean in $\mathcal{F}(\Omega(\bm x, R), \alpha)$ is also in $\mathcal{F}_{C^+}(\Omega(\bm x, R), \alpha)$.
    
    Fix $G \in \mathcal{G}(\Omega(\x, R), \alpha)$ and construct distribution $H$ as follows. 
    Let $k = |C|$, $s_j$ denote the $j$th largest element in $C$, and $s_j'$ be the smallest element in $S$ such that $s_j < s_j'$. $H$ agrees both pointwise and cumulatively with $G$ on $C$ and
    \begin{enumerate}
        \item If $s_1 > S_{\min}$, then $H(S_{\min}) = P_{G}[X < s_1]$.
        \item If $s_k < S_{\max}$, then $H(s_k') = 1 - P_{G}[X \leq s_k]$.
        \item $\forall j \in \{1, \ldots, k-1\}$, if $s_j' \neq s_{j+1}$, then $H(s_j') = P_{G}[X < s_{j+1}] - P_{G}[X \leq s_j]$.
    \end{enumerate}
    In words, Property 3 states that $H(s_j')$ is equal to the mass that $G$ assigns to support points above $s_j$ but below $s_{j+1}$. Because $G$ and $H$ agree pointwise and cumulatively on $C$, we have by assumption that $P_{G}[\Omega(\x, R)] = P_{H}[\Omega(\x, R)]$, which implies that $H \in \mathcal{G}_{C^+}(\Omega(\bm x, R), \alpha)$.  
    And by construction we also have that $E[H] \leq E[G]$. 
\end{proof}

\begin{theorem}
\label{thm:quantile_bound_restriction}
    Let $C = \{x_{(i)}\}$. For fixed $\x \in \Omega$ and $\alpha \in [0,1)$, 
    \begin{equation}
        B_{R_i}^*(\x) = \min\{E[F]: F \in \mathcal{F}_{C^+}(\Omega(\bm x, R_i), \alpha)\}.
    \end{equation}
\end{theorem}
\begin{proof}
    From Lemma~\ref{lem:support_refinement} it will suffice to show that for any two distributions $G, H \in \mathcal{F}$ agreeing cumulatively and pointwise on $C$, $P_G[\Omega(\x, R_i)] = P_H[\Omega(\x, R_i)]$. By Definition~\ref{def:quantile_order} and the fact that $H$ agrees cumulatively with $G$ at $x_{(i)} \in C$ we have
    \begin{equation}
    \begin{array}{rcl}
        P_G[\Omega(\x, R_i)] &=& P_G[\x \lesssim_{R_i} \Y]\\ 
        &=& P_G[x_{(i)} \leq Y_{(i)}]\\ 
        &=& P_H[x_{(i)} \leq Y_{(i)}]\\ 
        &=& P_H[\x \lesssim_{R_i} \Y]\\
        &=& P_H[\Omega(\x, R_i)]
    \end{array}.
    \end{equation}
\end{proof}

\begin{theorem}
    Let $c = (S_{\max} - S_{\min}) / (m-1)$. For each $\x \in \Omega$, $\alpha \in [0,1)$, and $\epsilon > 0$, there exists a polynomial-time $(c+\epsilon)$-approximation for calculating $B_{R_i}^*(\x)$.
\end{theorem}
\begin{proof}
    Fix $\x \in \Omega$. Let $C = \{x_{(i)}\}$, $C^- = \{S_{\min}, x_{(i)}\}$, and, by Definition~\ref{def:augmentation}, we have $C^+ = \{S_{\min}, x_{(i)}, x_{(i)}'\}$, where $x_{(i)}'$ is $S_{\max}$ if $x_{(i)} = S_{\max}$ or otherwise the least element of $S$ that is greater than $x_{(i)}$.    Theorem~\ref{thm:quantile_bound_restriction} establishes that there exists a distribution $F^* \in \mathcal{F}_{C^+}(\Omega(\x, R_i), \alpha)$ whose mean achieves $B_{R_i}^*(\bm x)$. And by Lemma~\ref{lem:mean_eps_close}, we also have that $\forall \varepsilon > 0, \exists G^* \in \mathcal{G}_{C^+}(\Omega(\x, R_i), \alpha): |E[F^*] - E[G^*]| < \varepsilon$. 
    
    Corresponding to each $G \in \mathcal{G}_{C^+}(\Omega(\x, R_i), \alpha)$ is a distribution $H_p \in \mathcal{F}_{C^-}$ where $H_p$ is equal to $G$, except that all mass at $G(x_{(i)}')$ is transferred to $x_{(i)}$ such that $P_{H_p}[x_{(i)}] = p$ and $P_{H_p}[S_{\min}] = 1-p$. Notice that, by construction, 
    \begin{equation}
    \label{eq:g_h_gap}
    \begin{array}{rcl}
        E[G] - E[H_p] &\leq& p (x_{(i)}' - x_{(i)}) \\
        &\leq& (x_{(i)}' - x_{(i)}) \\
        &\leq& \frac{S_{\max} - S_{\min}}{m-1} \\
    \end{array}.
    \end{equation}
    Also by construction and Definition~\ref{def:quantile_order}, $P_G[x_{(i)} \leq X_{(i)}] = P_{H_p}[x_{(i)} \leq X_{(i)}]$. Therefore, $H_p \in \mathcal{G}_{C^+}(\Omega(\x, R_i), \alpha)$.
    Take $H_{p^*}$ to be the distribution in $ \mathcal{F}_{C^-}$ corresponding to $G^*$. By letting $\varepsilon \rightarrow 0$ and by Inequality~\ref{eq:g_h_gap} and the definition of $G^*$ we have 
    \begin{equation}
        |E[F^*] - E[H_{p^*}]| \leq \frac{S_{\max} - S_{\min}}{m-1},
    \end{equation}
    where $E[H_p] = p x_{(i)}$. Thus, it remains only to approximate $p^*$ to within an additive factor of $\epsilon / x_{(i)}$.

    According to Definition~\ref{def:quantile_order} we know that $\y \in \Omega(\bm x, R_i) \Leftrightarrow x_{(i)} \leq y_{(i)}$, and in order to have $x_{(i)} \leq y_{(i)}$, given support $C^-$, it must be the case that there are at least $i$ occurrences of $x_{(i)}$ in $\y$, with the remaining elements of $\y$ being $S_{\min}$. Let $V_i$ be the event that any such $\y$ is drawn. For distribution $H_p$ we have, $H_p \in \mathcal{G}_{C^+}(\Omega(\x, R_i)) \Leftrightarrow P_{H_p}[V_i] > \alpha$ where $P_{H_p}[V_i] = 1-\texttt{Bin}(n-i-1; n, p)$ and  $\texttt{Bin}$ denotes the cumulative binomial distribution. 
    $\texttt{Bin}(n-i-1; n, p)$ must be decreasing in $p$ since it gives the probability that the number of successes is \emph{limited} to $n-i-1$, which implies that $P_{H_p}[V_i]$ is increasing in $p$. Of course $E[H_p]$ is also increasing in $p$. So to approximate $p^*$ we seek the smallest $p$ such that $P_{H_p}[V_i] > \alpha$, i.e. $H_p$ remains in $\mathcal{G}_{C^+}(\Omega(\x, R_i))$. 
    We perform binary search for $p^*$ on the interval $[0, 1]$ as follows. 
    \begin{enumerate}
        \item For interval $I = [a, b]$, let $|I| = b-a$, $I^- = [a, a+|I|/2]$, and $I^+ = [b-|I|/2, b]$.
        \item Initially choose $I = [0, 1]$ and $\delta = \epsilon / x_{(i)}$.
        \item While $|I| > \delta$, take $p' = a + |I|/2$ and
        \begin{enumerate}
            \item If $P_{H_{p'}}[V_i] < \alpha$, then set $I = I^+$
            \item Otherwise set $I = I^-$.
        \end{enumerate}
        \item Return $\min(I)$.
    \end{enumerate}
    
    At every stage of the algorithm, interval $I$ contains $p^*$. Yet by the end of the algorithm, $|I| \leq \delta$. Thus, by choosing $p = \min(I)$, we can be sure that $|p - p^*| \leq \epsilon / x_{(i)}$. 

    The algorithm runs for $O(\log(x_{(i)} / \epsilon))$ steps. And in each step, $P_{H_{p'}}[V_i]$ can be calculated to within machine precision in time polynomial in $n$ and $i$ using common numerical procedures. Thus, the overall algorithm is polynomial in the inputs $n$, $i$, $x_{(i)}$, and $\epsilon^{-1}$.
\end{proof}

\begin{theorem}
    Let $C = \{S_j \in S: \exists k \in [n]: x_k = S_j\}$. For fixed $\x \in \Omega$, 
    \begin{equation}
        B_{T_\ell}^*(\x) = \min\{E[F]: F \in  \mathcal{F}_{C^+}(\Omega(\bm x, T_\ell), \alpha)\}.
    \end{equation}
\end{theorem}
\begin{proof}
    From Lemma~\ref{lem:support_refinement} it will suffice to show that for any two distributions $G, H \in \mathcal{F}$ agreeing cumulatively and pointwise on $C$, $P_G[\Omega(\x, T_\ell)] = P_H[\Omega(\x, T_\ell)]$.
    According to Definition~\ref{def:lexi_low}, we have that $G \in \mathcal{G}(\Omega(\bm x, T_\ell), \alpha) \Leftrightarrow P_G[L_1(\x, \Y)] > \alpha$. But, for each $j \in [n]$ and any $\x, \y \in \Omega$, the events $(x_{(j)} < y_{(j)})$ and $ (x_{(j)} = y_{(j)})$ are mutually exclusive, so from Equation~\ref{eq:lexi_low} we have that for all $j \in [n]$
    \begin{equation}
    \label{eq:prob_lexi_greater}
    \begin{array}{l}
        P_G[L_j(\x, \Y)] = \\
        P_G[x_{(j)} < Y_{(j)}] + P_G\left[L_{j+1}(\x, \Y) ~|~ x_{(j)} = Y_{(j)}\right] P_G[x_{(j)} = Y_{(j)}]
    \end{array}.
    \end{equation}
    Therefore, applying Lemma~\ref{lem:cond_agreement}, it follows that $P_G[L_1(\x, \Y)] = P_H[L_1(\x, \Y)]$ so that 
    \begin{equation}
    \begin{array}{rcl}
        P_G[\Omega(\x, T_\ell)] &=& P_G[\x \leq_{T_\ell} \Y]\\ 
        &=& P_G[L_1(\x, \Y)]\\ 
        &=& P_H[L_1(\x, \Y)]\\ 
        &=& P_H[\x \leq_{T_\ell} \Y]\\
        &=& P_H[\Omega(\x, T_\ell)]
    \end{array}.
    \end{equation}
\end{proof}

\appendix 

\section{Supporting Probabilistic Results}

In this section we introduce some additional vector notation to ease exposition. Let $\x_{(:i)}$ and $\x_{(i:)}$ denote, respectively, the order statistics of $\x$ before and after the $i$th order statistic (inclusive), with $\x_{(1:)} = \x_{(:n)}$ denoting the entire vector $\x$ sorted in increasing order.

\begin{lemma}
\label{lem:cond_agreement}
    Let $\x \in \Omega$ be fixed and suppose that $G, G' \in \mathcal{F}$ agree both cumulatively and pointwise on $S_{\bm x}$. Then $\forall i \in [n]$ we have
    \begin{equation}
    \label{eq:cond_eq_agree}
        P_G[x_{(i)} = Y_{(i)} | \x_{(:i-1)} = \Y_{(:i-1)}] = P_{G'}[x_{(i)} = Y_{(i)} | \x_{(:i-1)} = \Y_{(:i-1)}],
    \end{equation}
    and
    \begin{equation}
    \label{eq:cond_lt_agree}
        P_G[x_{(i)} < Y_{(i)} | \x_{(:i-1)} = \Y_{(:i-1)}] = P_{G'}[x_{(i)} < Y_{(i)} | \x_{(:i-1)} = \Y_{(:i-1)}],
    \end{equation}
    where for each $H \in \mathcal{F}$, $P_H[x_{(1)} = Y_{(1)} | \x_{(:0)} = \Y_{(:0)}] \equiv P_H[x_{(1)} = Y_{(1)}]$ and $P_H[x_{(1)} < Y_{(1)} | \x_{(:0)} = \Y_{(:0)}] \equiv P_H[x_{(1)} < Y_{(1)}]$.
\end{lemma}
\begin{proof}
First note that since $G$ and $G'$ agree cumulatively and pointwise on $S_\x$ it is also the case that $\forall x \in S_\x$, $P_G[X < x] = P_{G'}[X < x]$. It can be shown~\cite{arnold2008first}[Thm. 3.3.1] that for fixed $i \in [n]$ and any $H \in \mathcal{F}$, $P_H[\X_{(:i)} = \x_{(:i)}] = I(\x_{(:i)}, i, n)$, where 
\begin{equation}
I(\x_{(:i)}, i, n) = C(i, n) \int\limits_{B_H(\x_{(:i)}, i)} D(\uu, i, n) d \uu,
\end{equation}
$C(i, n)$ and $D(\uu, i, n)$ are known but unimportant for our purposes, and
\begin{equation}
\begin{array}{l}
    B_H(\x_{(:i)}, i) \equiv \\ 
    \{(u_{1}, \ldots, u_{i}): \forall j \leq i, u_{j} \leq u_{j+1}, P_H[X < x_{(j)}] \leq u_{j} \leq P_H[X \leq x_{(j)}] \}
\end{array}.
\end{equation}
This implies that $I(\x_{(:i)}, i, n)$ depends on $H$ only through $B_H(\x_{(:i)}, i, n)$, and only at the points in $S_{\x(:i)}$. Thus, it is clear that $B_G = B_{G'}$, and therefore
\begin{equation}
\label{eq:point_equiv_g_gp}
    P_G[\X_{(:i)} = \x_{(:i)}] = P_{G'}[\X_{(:i)} = \x_{(:i)}].
\end{equation}
The validity of Equation~\ref{eq:cond_eq_agree} follows by observing that 
\begin{equation}
\label{eq:bayes}
\arraycolsep=1.4pt\def\arraystretch{2}
\begin{array}{rcl}
    P_G[x_{(i)} = Y_{(i)} | \x_{(:i-1)} = \Y_{(:i-1)}] &=& \frac{ P_G[\x_{(:i)} = \Y_{(:i)}]}{P_G[\x_{(:i-1)} = \Y_{(:i-1)}]} \\
    &=& \frac{ P_{G'}[\x_{(:i)} = \Y_{(:i)}]}{P_{G'}[\x_{(:i-1)} = \Y_{(:i-1)}]} \\
    &=& P_{G'}[x_{(i)} = Y_{(i)} | \x_{(:i-1)} = \Y_{(:i-1)}]
\end{array}
\end{equation}
Since $\x$ can be arbitrary, it also follows from Equation~\ref{eq:point_equiv_g_gp} that
\begin{equation}
    \begin{array}{rcl}
        P_G[x_{(1)} < Y_{(1)}] &=& \sum\limits_{y \in S, y > x_{(1)}} P_G[Y_{(1)} = y] \\
        &=& \sum\limits_{y \in S, y > x_{(1)}} P_{G'}[Y_{(1)} = y] \\
        &=& P_{G'}[x_{(1)} < Y_{(1)}]
    \end{array}.
\end{equation}

Now for fixed $i \in \{2, \ldots, n\}$ and $H \in \mathcal{F}$ define $S' \subseteq S$ such that $y \in S'$ iff $x_{(i)} < y \leq S_{\max}$. Notice that, by construction, 
\begin{equation}
    \bigcup_{y \in S'} \{t: P_H[X < y] \leq t \leq P_H[X \leq y]\} = \{t: P_H[X < x_{(i)}] < t \leq 1\}.
\end{equation}
Define
\begin{equation}
    I_H'(t, \x_{(:i-1)}, i, n) = C(i, n) \int\limits_{B_H(\x_{(:i-1)}, i-1)} D(\uu_{(:i-1)} \oplus \{t\}, i, n) d \uu_{(:i-1)},
\end{equation}
where operator $\oplus$ indicates concatenation, and again the only dependence on H is through $B_H$. We have
\begin{equation}
\label{eq:partial_joint}
\arraycolsep=1.4pt\def\arraystretch{2}
\begin{array}{rcl}
    P_H[X_{(i)} > x_{(i)}, \X_{(:i-1)} = \x_{(:i-1)}] &=& \sum\limits_{y \in S'} I_H(\x_{(:i-1)} \oplus \{y\}, i, n) \\
    &=& \sum\limits_{y \in S'}  \int\limits_{B_H'(y)} I_H'(t, \x_{(:i-1)}, i, n) dt \\
    &=& \int\limits_{B_H''(x_{(i)}, t)} I_H'(t, \x_{(:i-1)}, i, n) dt \\
\end{array},
\end{equation}
where
\begin{equation}
    B_H'(y) = \{t: P_H[X < y] \leq t \leq P_H[X \leq y]\},
\end{equation}
and
\begin{equation}
    B_H''(x_{(i)}, t) = \{t: P_H[X \leq x_{(i)}] < t \leq 1\}.
\end{equation}
Ultimately the only dependence on $H$ in Equation~\ref{eq:partial_joint} is through $B_H$ and $B_H''$. And, because $G$ and $G'$ agree cumulatively on $S_{\bm x}$, it must be the case that $B_G'' = B_{G'}''$. Therefore, it is clear that 
\begin{equation}
    P_G[x_{(i)} < Y_{(i)}, \x_{(:i-1)} = \Y_{(:i-1)}] = P_{G'}[x_{(i)} < Y_{(i)} , \x_{(:i-1)} = \Y_{(:i-1)}],
\end{equation} 
and the validity of Equation~\ref{eq:cond_lt_agree} follows by similar argument as in Equation~\ref{eq:bayes}.
\end{proof}

\printbibliography

\end{document}